\documentclass[runningheads]{llncs}

\usepackage{amsfonts,amssymb,amscd,amsmath,euscript,enumerate,verbatim,calc}
\usepackage{amssymb,amsmath}

\usepackage{algorithm2e}
\usepackage{xcolor}
\newcommand{\new}[1]{{\color{black}{#1}}}

\begin{document}
\title{Modification tolerant signature schemes: location and correction}
\author{Thais Bardini Idalino\thanks{Funding granted from CNPq-Brazil [233697/2014-4] and OGS.} \and Lucia Moura \and Carlisle Adams}
\institute{University of Ottawa, Ottawa, Canada\\
\email{\{tbardini,lmoura,cadams\}@uottawa.ca}}
	
\maketitle
\makebox[\linewidth]{\small October 2019}

\begin{abstract}
This paper considers malleable digital signatures, for situations where  data is modified after it is signed.
They can be used in applications where either the data can be modified (collaborative work), or the data must be modified (redactable and content extraction signatures) or we need to know which parts of the data have been modified (data forensics). A \new{classical} digital signature is valid for a message only if the signature is authentic and not even one bit of the message has been modified. We propose a general framework of modification tolerant signature schemes (MTSS), which can provide either location only or both location and correction, for modifications in a signed message divided into $n$ blocks. This general scheme uses a set of allowed modifications that must be specified. We present an instantiation of MTSS with a tolerance level of $d$, indicating modifications can appear in any set of up to $d$ message blocks. This tolerance level $d$ is needed in practice for parametrizing and controlling the growth of the signature size with respect to the number $n$ of blocks; using combinatorial group testing (CGT) the signature has size $O(d^2 \log n)$ which is close to the \new{best known} lower bound \new{of $\Omega(\frac{d^2}{\log d} (\log n))$}. There has been work in this very same direction using CGT by Goodrich et al. (ACNS 2005) and Idalino et al. (IPL 2015). Our work differs from theirs in that in one scheme we extend these ideas to include corrections of modification with provable security, and in another variation of the scheme we go in the opposite direction and guarantee privacy for redactable signatures, in this case preventing any leakage of redacted information.
\end{abstract}


\section{Introduction}

\new{Classical} digital signature schemes (CDSS) are used to guarantee that a document was created by the sender (authenticity) and has not been modified along the way (integrity); they also help to prevent the signer from claiming s/he did not sign a message (non-repudiation). The verification algorithm has a boolean output: a successful outcome is achieved if and only if both the signature is valid and the document has not been modified. 
In this paper, we consider more general digital signature schemes which we call {\em modification-tolerant signature schemes} (MTSS), which go beyond the ability of 
{\bf detecting} modifications provided by CDSS, and have the ability of  {\bf locating} modifications or {\bf locating and correcting} modifications. 
We discuss two types of modification-tolerant signature schemes: a general MTSS that allows the location of modified blocks of the data, and an MTSS with \emph{correction capability}, that allows the correction of the modified blocks, recovering the original message. We give three instantiations of the scheme for the purpose of location, correction, and redaction.

In which situations can modifications be allowed or even desired in the context of digital signatures?
One situation where modifications are desirable is the so-called redactable signatures~\cite{mtss:Johnson}, also called content extraction signatures~\cite{mtss:Steinfeld}.
In redactable signature schemes~\cite{mtss:Johnson,mtss:PohlsThesis,mtss:Steinfeld}, some blocks of a document are redacted (blacked out) for privacy purposes, without interacting with the original signer and without compromising the validity of the signature. Related to these are the ``sanitizable'' signatures introduced by Ateniese et al.~\cite{mtss:Ateniese}
which are also used for the purpose of privacy of (parts of) the original message, but the modifications are done by a semi-trusted third party (the sanitizer) who can modify blocks of the document and sign the modified version, without the need of intervention by the original signer. 
Thus, in both redactable and sanitizable signature schemes, the privacy of the (redacted or modified parts of the) original message must be preserved. For MTSS, privacy of the original data that has been modified is not required. An MTSS with only location capability that guarantees privacy can be used for implementing redactable signatures, but that is not an intrinsic requirement for MTSS. Indeed, as pointed out in~\cite{mtss:PohlsThesis} the scheme provided in~\cite{mtss:thaisIPL} does not meet standard privacy requirements. In the case of MTSS with correction capability, privacy cannot be guaranteed by definition, since the method \new{permits} the recovery of the original document.

A different scenario where a moderate amount of modification is desirable involves collaborative work. The authors of a document can use MTSS to allow further modifications as long as the original parts can be identified as their own. Other collaborators may apply modifications to the original document and append the original document's signature, which provides information about which blocks were modified, as well as a guarantee of integrity of the original blocks. MTSS can separate the original blocks from the modified ones, while correction capability further provides retrieval of the original version of the document.

Locating modifications has also been considered in situations where modifications are not desired, but their localization is used as a mechanism to mitigate damage.  Indeed, in the context of message authentication codes for data structures, Goodrich et al.~\cite{mtss:forensics} propose a message authentication code (MAC)  with modification locating capabilities. They propose their use in data forensics applications since the identification of which information was modified can be used to identify the perpetrator of the crime (for example: the salary of a person or the grade of a student was modified on a database). In~\cite{mtss:thaisIPL}, in the context of MTSS with only location capability, the authors mention the advantage of being able to guarantee the integrity of part of the data instead of the all-or-nothing situation given by a CDSS boolean outcome. For example, the integrity of 95\% of a document or database may contain enough information needed for a specific application, whereas it would have been considered completely corrupted and unusable in the case of CDSS. In the case of MTSS with correction capability, we can go beyond mitigating the damage, and actually recover the originally signed document.

The mechanism behind the MTSS schemes instantiated here, like in \cite{mtss:forensics,mtss:thaisIPL}, is the use of cover-free families (CFF) in the same way as it is traditionally employed in combinatorial group testing. Combinatorial group testing has been used in cryptography in the context of digital signatures~\cite{mtss:zaverucha}, broadcast communication~\cite{mtss:blacklistingStinson,mtss:broadcastauth}, and many others. The main idea is to test $t$ groups, which are subsets of the $n$ blocks, together (with enough redundancy and intersections between groups), and each group is used to produce a distinct hash. The tuple of $t$ hashes is then signed and provided that no more than $d$ blocks were modified, it is possible to identify precisely which blocks have been modified. The main efficiency concern is the compression ratio: the order of growth of $n/t$ as $n$ grows, for a fixed modification tolerance $d$. Using cover-free families, it is possible to achieve a compression ratio of $O(n/(d^2 \log n))$, which is not much worse than the $O(n)$ compression ratio given by modification intolerant schemes such as CDSS.

\noindent
\textbf{Our contributions:} 
In the present paper, we propose a general framework for MTSS, and a specific MTSS scheme for modification correction and another for redacted signatures. Both schemes are based on a modification tolerance $d$, indicating the maximum number of modifications or redactions. The security of the schemes rely only on an underlying \new{classical} digital signature scheme and on collision resistant hash functions; therefore the schemes can be used in the context of postquantum cryptography as long as these underlying building blocks are postquantum secure.

First, we extend methods that use cover-free families for modification location~\cite{mtss:forensics,mtss:thaisIPL} to further provide modification correction (Scheme 2 in Section~\ref{mtss:MTSSCFF}), provided that the size of the blocks are small enough, say bounded by a constant $s$ where an algorithm with $2^s$ steps can run in an acceptable amount of time. In short, the localization of the modified blocks is done using CFF similarly to \cite{mtss:forensics,mtss:thaisIPL} with an extra constraint on blocks of size at most $s$, and an exhaustive search is used to correct the block to a value that makes the concatenation of a specific group of blocks match its original hash. The assumption that a collision resistant hash function does not cause collisions for messages of small size up to $s$ is not only reasonable, but can be tested before employing it in the method.

Second, we propose a variation of the scheme for modification location to ensure total privacy of the modified blocks in order to extend it for the purpose of redactable signatures (Scheme 3 in Section~\ref{mtss:SecRedactable}). In this case, a block modification can be a redaction to hide private information. Unlike the modification correction scheme, this scheme does not need a restriction on block size.

\noindent
\textbf{Paper Organization:} 
In Section~\ref{mtss:SecRelated}, we discuss related work. In Section~\ref{mtss:SecFramework}, we give a general framework for modification tolerant signature schemes with and without modification correction. In Section~\ref{mtss:SecSchems}, we instantiate the schemes for location and/or correction that allows any modifications in up to $d$ blocks using cover-free families. In Section~\ref{mtss:SecSecurity}, we prove the unforgeability of the schemes under the adaptive chosen message attack. In Section~\ref{mtss:SecRedactable},  we extend and instantiate the modification location scheme for redactable signatures, and prove it guarantees total privacy. In Section~\ref{mtss:implementation}, we discuss implementation issues of the schemes proposed in this paper. In Section~\ref{mtss:param}, we consider the relationship of parameters and the impact on the size of the signature and the ability to locate and correct modifications.
A conclusion is given in Section~\ref{mtss:conclusion}.

%%%%%%%%%%%%%%%%%%%%%%%%%%%%%%%%%%%%%%%%%%%%%%%%%%%%%%%%%%

\section{Related Work}\label{mtss:SecRelated}
The idea of error detection and correction of corrupted data is well established in coding theory. Combinatorial group testing can help locating where errors (modification of data) occurred, which can be considered an intermediate goal, which is stronger than error detection and more efficient than error correction.
In the context of cryptography, localization of modified portions of data has been studied in the context of hash functions~\cite{mtss:locHash,mtss:locHash2}, digital signatures~\cite{mtss:Biyashev2012,mtss:thaisIPL}, and message authentication codes~\cite{mtss:DBMAC,mtss:forensics}.
Correction of modifications is proposed by some schemes, but they have severe limitations on the correction capability~\cite{mtss:Biyashev2012,mtss:DBMAC}.
Schemes such as the ones in~\cite{mtss:locHash,mtss:locHash2,mtss:forensics,mtss:thaisIPL} use techniques related to cover-free families to generate redundant information, which is used later to locate the modified portions of the data. The benefit of cover-free families is that they require a small amount of redundant data for a fixed threshold $d$ in location capability.

The methods that provide location at the hash level~\cite{mtss:locHash,mtss:locHash2} are based on superimposed codes, which are binary codes equivalent to $d$-cover-free families ($d$-CFF). These codes are used in two steps of the process: to generate a small set of hash tags from an input message, and to verify the integrity of the message using these tags. Because of the $d$-CFF property, the scheme allows the identification of up to $d$ corrupted data segments~\cite{mtss:locHash,mtss:locHash2}. 
A digital signature scheme with location of modifications is presented in~\cite{mtss:thaisIPL}. Their scheme uses $d$-CFF matrices to generate a digital signature scheme that carries extra information for location of modifications. In this case, data is divided into $n$ blocks which are concatenated and hashed according to the rows of the $d$-CFF. This set of hash values is signed using any \new{classical} digital signature algorithm, and both the signature and the set of hashes form the new signature of the message. The verification algorithm uses this information to precisely locate up to $d$ blocks that contain modified content~\cite{mtss:thaisIPL}. This is presented in detail in Section~\ref{mtss:MTSSCFF} as Scheme 1.
In~\cite{mtss:forensics}, the authors propose a solution for locating the items that were modified inside a data structure. They compute a small set of message authentication tags based on the rows of a $d$-disjunct matrix, which is equivalent to $d$-CFFs. These tags are stored within the topology of the data structure, and an auditor can later use these tags to determine exactly which pieces of data were modified and perform forensics investigation.

The methods described above have a common approach: they all compute extra integrity information based on combinatorial group testing techniques and use this information to locate up to $d$ modified pieces of data. In our work, we show that if these pieces are small enough, it is possible to actually correct them to the originally-signed values. In the literature, we find signature and message authentication schemes that provide the correction of modifications, such as in~\cite{mtss:Biyashev2012,mtss:DBMAC}. These schemes use different approaches than the ones previously discussed, and the capacity of location and correction is very limited. In~\cite{mtss:Biyashev2012} only a single error can be corrected, while in~\cite{mtss:DBMAC} the data is divided into blocks, and only one bit per block can be corrected.

Error correcting codes ($d$-ECC) are related to this work, as we could use them to provide authentication and correction of $d$ modifications as we explain next. For a message $m$ formed by blocks of up to $\log_2 q$ bits each,
sign it using a CDSS and compute the corresponding codeword $c_m$ using a $d$-ECC over alphabet $q$. Send the codeword $c_m$ and the signature $\sigma$. 
The receiver obtains $c'$, which may differ from $c_m$ in at most $d$ blocks, and by decoding $c'$, obtains $c_m$ and thus the original message $m$. Then, the verifier can check its authenticity and integrity by verifying $(m,\sigma$). \new{We call this scheme \emph{Scheme E}}. In Section~\ref{section:comparison}, we compare \new{Scheme E} and Scheme 2 using $d$-CFF, and find that they have very similar compression ratios. Scheme 2 has the advantage of giving more information in the failing case, of being more efficient when we do not need correction, and of being a simple variation of Schemes 1 and 3, facilitating comparison among the different approaches. 

There is also a solution proposed in~\cite{mtss:Barreto} for location of modifications in images with digital signatures, using a different approach than $d$-CFFs.
Their scheme consists \new{of} partitioning the image into $n$ blocks, and generating one digital signature per block (with dependency on neighbours to avoid attacks). 
Although we can use signature schemes with very small signature sizes, the total number of signatures in this scheme is linear with the number of blocks $n$, while in~\cite{mtss:thaisIPL}, for example, the amount of extra information produced is logarithmic in $n$ because of the $d$-CFF.

Redactable or content extraction digital signatures~\cite{mtss:Derler,mtss:Johnson,mtss:Steinfeld} are used when the owner of a document signed by another party, needs to show parts of that document to a third party, but wants to keep parts of the document private. In Section~\ref{mtss:SecRedactable}, we give a variation of an MTSS scheme that implements redactable signatures. More on related work involving redactable signatures is discussed in that section.

%%%%%%%%%%%%%%%%%%%%%%%%%%%%%%%%%%%%%%%%%%%%%%%%%%%%%%%%%%

\section{Framework for modification-tolerant digital signatures}\label{mtss:SecFramework}

\subsection{Classical Digital Signature Schemes (CDSS)}
Classical digital signature schemes are based on public-key cryptography and consist of three algorithms: {\sc KeyGeneration},  {\sc Sign}, and {\sc Verify}. We consider that any document or piece of data to be signed is a sequence of bits called a {\em message}.

\begin{definition}
	A classical digital signature scheme (CDSS) is a tuple $\Sigma$ of three algorithms:
	\begin{itemize} 
		\item {\sc KeyGeneration}$(\ell)$ generates a pair of keys (secret and public) ($SK, PK$) for a given security parameter $\ell$.
		\item {\sc Sign}$(m, SK)$ receives the message $m$ to be signed, the secret key $SK$, and outputs the signature $\sigma$.
		\item {\sc Verify}$(m, \sigma, PK)$ takes as input a message $m$, signature $\sigma$, and public key $PK$. Using $PK$, outputs $1$ if the pair ($m, \sigma$) is valid (as in Definition~\ref{mtss:def:validCDSS}), and outputs 0 otherwise.
	\end{itemize}
\end{definition}	

\begin{definition}\label{mtss:def:validCDSS}
	Let $\Sigma$ be a CDSS as defined above, and let $(SK, PK)$ be a pair of secret and public keys. A pair of message and signature $(m, \sigma)$ is \emph{valid} if $\sigma$ is a valid output\footnote{We use ``valid output" instead of $\Sigma$.{\sc Sign}$(m,SK) = \sigma$ because the signing algorithm does not need to be deterministic.} of $\Sigma$.{\sc Sign}$(m,SK)$.
\end{definition}

Generally speaking, we say CDSS is \emph{unforgeable} if a signature that verifies using $PK$ can only be generated by the signer who has $SK$.
In more detail, we consider the model of \emph{adaptive chosen message attack}, where the attacker $\mathcal{A}$ adaptivelly chooses a list of messages $m_1, \ldots, m_q$, and requests the respective signatures $\sigma_1, \ldots, \sigma_q$ from an oracle $\mathcal{O}$. Given this information, we say the attacker performs an \emph{existential forgery} if he is able to produce a valid signature $\sigma$ on a new message $m$ chosen by him, $m \neq m_i, 1 \leq i \leq q$ \cite{mtss:menezesBook,mtss:stinsonBook}.
Because with unlimited time an adversary can perform all possible computations, we limit the computational power of the attacker by requiring an efficient algorithm to be one that runs in probabilistic polynomial time (PPT). Moreover, we say a function $f$ is \emph{negligible} if for every polynomial $p$ there exists an $N$ such that for all $n > N$ we have that $f(n)< 1/p(n)$~\cite{mtss:katzLindell}.

\begin{definition}[Unforgeability]\label{mtss:defsecurity}
	A digital signature scheme is existentially unforgeable under an adaptive chosen message attack if there is no PPT adversary $\mathcal{A}$ that can create an existential forgery with non-negligible probability.
\end{definition}

\subsection{Description of MTSS}

Now we define the algorithms of the modification tolerant signature scheme (MTSS).
We assume that the message $m$ to be signed has size $|m|$ in bits, and is split into $n$ \emph{blocks}, not necessarily of the same size. A message split into blocks is represented as a sequence $m = (m[1], \ldots, m[n])$, where $m[i]$ represents the $i$th block of message $m$. For two messages $m$ and $m'$, each split into $n$ blocks, we denote diff$(m,m') = \{i \in \{1, \ldots, n\}: m[i] \neq m'[i]\}$. An \emph{authorized modification structure} is a collection $\mathcal{S} \subseteq P(\{1, \ldots, n\})$, where $P$ is the power set of $\{1, \ldots, n\}$, that contains each set of blocks that the signer allows to be modified. The idea is that if $\sigma$ is a valid signature for $m$, then $\sigma$ is a valid signature for $m'$ if and only if diff$(m,m') \in \mathcal{S}$. Of course, to prevent inconsistencies we must have $\emptyset \in \mathcal{S}$. Indeed, an MTSS is equivalent to a CDSS if $\mathcal{S} = \{\emptyset\}$. A more general example is $\mathcal{S} = \{\emptyset, \{2\}, \{3\}, \{5\}, \{2,3\}, \{2,5\}, \{3,5\}\}$ for $n = 5$ blocks, which specifies blocks $1$ and $4$ \new{cannot} be changed and any change of at most two other blocks is allowed. The authorized modification structure is used to provide flexibility of modifications while providing control for the signer. In practice, we do not expect $\mathcal{S}$ to be stored explicitly, but instead to be implicitly enforced by the scheme.

\begin{definition}\label{mtss:def:MTSS}
	A modification tolerant signature scheme (MTSS) for authorized modification structure $\mathcal{S} \subseteq P(\{1, \ldots, n\})$ on messages with $n$ blocks is a tuple $\Sigma$ of three algorithms: 
	\begin{itemize}
		\item {\sc MTSS-KeyGeneration}$(\ell)$ generates a pair of secret and public keys ($SK, PK$) for a given security parameter $\ell$.
		\item {\sc MTSS-Sign}$(m, SK)$ receives the message $m$ to be signed, the secret key $SK$, and outputs the signature $\sigma$.
		\item {\sc MTSS-Verify}$(m, \sigma, PK)$ takes as input a message $m$, signature $\sigma$, and public key $PK$. Outputs $(1,I)$ if ($m, \sigma$) is valid for modification set $I$ (as in Definition~\ref{mtss:validsig}), and outputs 0 otherwise.
	\end{itemize}
\end{definition}	

\begin{definition}\label{mtss:validsig}
	Let $\Sigma$ be an MTSS for authorized modification structure $\mathcal{S} \subseteq P(\{1, \ldots, n\})$, and let $(SK, PK)$ be a pair of secret and public keys. A pair $(m, \sigma)$ of message and signature is \emph{valid} if there exists $m'$ such that $\sigma$ is a valid output\footnote{We use ``valid output" instead of $\Sigma$.{\sc MTSS-Sign}($m', SK$) = $\sigma$ because the signing algorithm does not need to be deterministic.} of $\Sigma$.{\sc MTSS-Sign}($m', SK$) and diff$(m,m') \in \mathcal{S}$. In this case, we say $(m, \sigma)$ is valid for modification set $I$ (where $I = \text{diff}(m, m')$).
\end{definition}

The definition of unforgeability of an MTSS scheme is exactly like Definition~\ref{mtss:defsecurity}, but the existential forgery now needs to produce a valid signature as given in Definition~\ref{mtss:validsig}. \new{This is in alignment with the notions introduced in~\cite{mtss:Derler}.} 

\begin{definition}
	An MTSS $\Sigma$ for authorized modification structure $\mathcal{S}\subseteq P(\{1, \ldots, n\})$ has \emph{correction capability} if it is a tuple $\Sigma$ of four algorithms, which, in addition to the algorithms in Definition~\ref{mtss:def:MTSS}, has the following algorithm:
	\begin{itemize}
		\item {\sc MTSS-Verify\&Correct}($m, \sigma, PK$): takes as input a message $m$, signature $\sigma$, and public key $PK$. Outputs a pair $(ver, m')$ where:
		\begin{enumerate}
			\item $ver = \Sigma$.{\sc MTSS-Verify}$(m, \sigma, PK)$.
			\item $m'$ is a message with $m' \neq \lambda$ (the corrected message) if $ver = (1, I)$, $I = \text{diff}(m, m')$, and $(m',\sigma)$ is a valid pair for modification set $I' = \emptyset$; in all other cases $m' = \lambda$, which indicates failure to correct.
		\end{enumerate}
	\end{itemize}
\end{definition}	

Location of modifications with MTSS would be trivial \new{for any} $\mathcal{S} \subseteq P(\{1, \ldots, n\})$ if the signer simply produced $\sigma$ as a tuple of $n$ signatures, one for each block. However, this would be extremely inefficient for a large number of blocks $n$. We must reconcile the objectives of having a large $\mathcal{S}$ and having a compact signature. Of course, the signature size depends on the security parameter, but once this is fixed, we would like the signature to grow moderately as a function of $n$. This motivates the following definition of compression ratio.

\begin{definition}
	An MTSS $\Sigma^n$ for messages with $n$ blocks and signature $\sigma$ with $|\sigma| \leq s(n)$ has compression ratio $\rho(n)$ if $\frac{n}{s(n)}$ is $O(\rho(n))$. 
\end{definition}	

The compression ratio measures how efficient our signature is with respect to the trivial scheme of keeping one signature per block, with $\rho(n) = O(1)$, supporting $\mathcal{S} = P(\{1, \ldots, n\})$. \new{Classical} signatures have $\rho(n) = n$ (best possible), but $\mathcal{S} = \{\emptyset\}$. 
In the next section we present a tradeoff, where $\rho(n) = \frac{n}{\log n}$  and $\mathcal{S}$ is the set of all sets of up to $d$ modifications, for fixed $d$. Indeed, when using cover-free families it is possible to have a compression ratio of $O(\frac{n}{d^2 \log n})$ using known CFF constructions~\cite{mtss:PR}, while \new{a} lower bound \new{on $s(n)$}~\cite{mtss:furedi} tells us we cannot \new{have $\rho(n)$ larger} than $\Theta(\frac{n}{(d^2/\log d) \log n})$.

\section{$d$-Modification Tolerant MTSS Based on Combinatorial Group Testing}\label{mtss:SecSchems}
Here we propose an MTSS that allows the modification of any set of up to $d$ blocks in the message, for a constant $d$ which we call a tolerance level. In other words, the authorized modification structure is $\mathcal{S} = \{T \subseteq \{1, \ldots, n\}: |T| \leq d \}$. To obtain a compact signature size, we rely on combinatorial group testing, which we summarize in Section~\ref{mtss:sec:CFFs} before we describe the scheme in Section~\ref{mtss:MTSSCFF}. Similar modification location methods based on combinatorial group testing have been proposed in \cite{mtss:forensics,mtss:thaisIPL}, and the instantiation we propose for the first three algorithms of MTSS (Section~\ref{mtss:MTSSCFF}) are based on~\cite{mtss:thaisIPL}. Our new contributions in this section include proof of security for the MTSS scheme based on cover-free families and the addition of error correction capability by proposing algorithm {\sc MTSS-Verify\&Correct} that corrects modifications in this context. 

\subsection{Cover-Free Families and Group Testing}\label{mtss:sec:CFFs}
Combinatorial group testing deals with discovering $d$ defective items in a set of $n$ items, via testing various groups of items for the presence of defects in each group. In nonadaptive combinatorial group testing, the groups are determined before the testing process starts. For the problems considered in this paper, a modified block is a defective item, and the groups are sets of blocks combined and hashed together. In our case, we must use nonadaptive combinatorial group testing, as the tests are generated at time of signing, while verification is done later. Cover-free families allow for a small number of groups that help to identify the modified blocks.

Recall that a permutation matrix of dimension $l$ is an $l \times l$ matrix that is obtained from permuting rows of the identity matrix.

\begin{definition}
	A $d$-cover-free family ($d$-CFF) is a $t \times n$ binary matrix $\mathcal{M}$ such that any set of $d+1$ columns contains a permutation submatrix of dimension $d+1$. 
\end{definition}

Each column of $\mathcal{M}$ corresponds to a block of the message, and each row corresponds to a group of blocks that will be tested together. A test fails if a group contains a modified block and passes if every block in a group is unchanged. 
The definition of $d$-CFF guarantees that for any column index $j$ and any other set of $d$ column \new{indices} $j_1, \ldots, j_d$, there will be a row $i$ s.t. $\mathcal{M}_{i,j} = 1$, and $\mathcal{M}_{i,j_1} = \ldots =  \mathcal{M}_{i,j_d} = 0$. In other words, for any unchanged block, there exists a test that contains that block but none of the up to $d$ modified blocks.

Figure~\ref{mtss:fig:cff} shows an example \new{of} how to use a $2$-CFF($9, 12$) to test a message with $12$ blocks (represented by the columns) by testing $9$ groups of blocks (the rows). Every unchanged block is in at least one group that passes the test (with result 0), and every group that failed the test (result \textbf{1}) \new{contains} at least one modified block, which in this example are blocks $3$ and $12$. We note that column ``result" is the bitwise-or of the columns corresponding to modified blocks 3 and 12.

\begin{figure}
	\begin{center}
		\begin{tabular}{c}
			\begin{tabular}{l|cccccccccccc|l}
				blocks & 1 & 2 & 3 & 4 & 5 &6 &7 &8 &9 &10& 11& 12& \\ \hline
				& \checkmark & \checkmark& X & \checkmark & \checkmark &\checkmark& \checkmark & \checkmark & \checkmark & \checkmark & \checkmark & X & result:\\ \hline
				$t_1$  & 1 & 0 & 0 &1  & 0 & 0 & 1 &0 &0 &1 &0 & 0 & 0\\
				$t_2$  & 1 & 0 & 0 & 0 &1  & 0 & 0 & 1 & 0 &  0 &  1&  0 &0\\
				$t_3$  & 1 & 0 & 0 & 0 & 0 & 1 & 0 &0 &1 &0 & 0 & 1 &\textbf{1}\\
				$t_4$  & 0 & 1 & 0 & 1 & 0 & 0 & 0 &0 &1 &0 & 1 & 0 & 0\\
				$t_5$  & 0 & 1 & 0 & 0 & 1 & 0 & 1 & 0 & 0 & 0 &  0 & 1&\textbf{1}\\
				$t_6$  & 0 & 1 & 0 & 0 & 0 & 1 & 0 & 1 & 0 & 1 & 0 &  0& 0\\
				$t_7$  & 0 & 0 & 1 & 1 & 0 & 0 & 0 & 1 & 0 & 0 & 0 & 1 & \textbf{1}\\
				$t_8$  & 0 & 0 & 1 & 0 & 1 & 0 & 0 & 0 & 1 & 1 & 0 & 0 & \textbf{1}\\
				$t_9$  & 0 & 0 & 1 & 0 & 0 & 1 & 1 &0 & 0 & 0 & 1 & 0  & \textbf{1}\
			\end{tabular}\\
		\end{tabular}
	\end{center}
	\caption{Example of a $2$-CFF($9, 12$).}
	\label{mtss:fig:cff}
\end{figure}

The next theorem is important for the efficiency of {\sc MTSS-Verify\&Correct}. Indeed, it ensures that for each modified block, we can find a test that contains it together with only other unchanged blocks. Therefore, an exhaustive trial-and-error can be used to guess this block until the hash of this group of blocks matches the original hash.

\begin{theorem}\label{mtss:uniqueT}
	Let $j_1, \ldots, j_d$ be the column \new{indices} that represent invalid elements. There is a row $i$ such that $\mathcal{M}_{i,j_1} = 1, \mathcal{M}_{i,j_2} = \ldots = \mathcal{M}_{i,j_d} = 0$.
\end{theorem}
\begin{proof}
	Since $\mathcal{M}$ is a $d$-CFF, it is also a $(d-1)$-CFF, therefore one of the rows in the permutation submatrix indexed by $j_1, \ldots, j_d$ is as stated. 
\end{proof}

\subsection{Description of $d$-Modification Tolerant Signature Scheme}\label{mtss:MTSSCFF}

The main idea of a $d$-modification tolerant signature scheme is to sign a message split into blocks by concatenating the hashes of these blocks according to a $d$-CFF matrix. This allows us to locate up to $d$ modified blocks in the signed message, and correct these modifications. In this context, we represent a concatenation of two strings $a$ and $b$ as $a||b$, and $\lambda$ represents an empty string.

\begin{definition}
	A $d$-modification tolerant signature scheme ($d$-MTSS) is an MTSS with authorized modification structure $\mathcal{S} = \{T \subseteq \{1, \ldots, n\}: |T| \leq d \}$.
\end{definition}
We now give an instantiation of $d$-MTSS using $d$-cover-free families based on~\cite{mtss:thaisIPL}.

\vspace{10pt}
\noindent
\textbf{Scheme 1: A $d$-Modification Tolerant Signature Scheme}
\vspace{5pt}

The scheme requires an underlying CDSS $\Sigma$, a public hash function $h$, and a $d$-CFF($t,n$) matrix $\mathcal{M}$.
The algorithms are given next:
\begin{itemize}
	\item {\sc MTSS-KeyGeneration}$(\ell)$: generates a key pair ($SK, PK$) using algorithm \\$\Sigma$.{\sc KeyGeneration}$(\ell)$.
	\item {\sc MTSS-sign}($m, SK$): Takes as input a secret key $SK$ and a message \\$m = (m[1], \ldots, m[n])$, and proceeds as follows.
	\begin{enumerate}
		\item Calculate $h_j = h(m[j]), 1 \leq j \leq n$.
		\item For each $1 \leq i \leq t$, compute $c_i$ as the concatenation of all $h_j$ such that $\mathcal{M}_{i,j} = 1, 1 \leq j \leq n$. Set $T[i] = h(c_i)$.
		\item Compute $h^* = h(m)$ and set $T = (T[1], T[2], \ldots, T[t], h^*)$.
		\item Calculate $\sigma' = \Sigma$.{\sc Sign}($T, SK$). Output signature $\sigma = (\sigma', T)$.
	\end{enumerate}
	\item {\sc MTSS-verify}($m, \sigma, PK$): takes as input a message $m = (m[1], \ldots, m[n])$, signature $\sigma = (\sigma', T)$ for $T = (T[1], T[2], \ldots, T[t], h^*)$, and public key $PK$, \new{and} proceeds as follows.
	\begin{enumerate}
		\item Ensure that $\Sigma$.{\sc Verify}$(T, \sigma', PK) = 1$, otherwise stop and output $(0, -)$.
		\item Check if $h^* = h(m)$. Stop and output $(1, \emptyset)$ if that is the case, continue otherwise.
		\item Use $\mathcal{M}$ and $m$ and do the same process as in steps 1 and 2 of {\sc MTSS-sign} to produce hashes $T'[1], \ldots, T'[t]$.
		\item Start with an empty set $V$, and for each $1 \leq i \leq t$ such that $T[i] = T'[i]$, compute the set of indices of unmodified blocks $V_i = \{j: \mathcal{M}_{i,j} = 1\}$, and accumulate these values in the set of all indices of unmodified blocks $V = V \cup V_i$. Compute $I = \{1, \ldots, n\}\setminus V$. If $|I| \leq d$ output $(1,I)$, else output $(0,I)$. 
	\end{enumerate}
\end{itemize}

The correctness of the scheme is shown next.

\begin{theorem}\label{mtss:verifyCorrect} Consider a valid signature $\sigma$ generated by {\sc MTSS-Sign} for a message $m$ and key pair ($SK, PK$), and let $m'$ be a possibly modified message with $|\text{diff}(m, m')| \leq d$. Then, {\sc MTSS-Verify}($m', \sigma, PK$) = $(1, \text{diff}(m, m'))$.
\end{theorem}
\begin{proof}
	Since $\sigma$ is valid, {\sc MTSS-Verify}($m', \sigma, PK$) does not stop at step 1. If $m=m'$, then $h(m) = h(m')$ and the algorithm will stop in step 2, with $(1, \text{diff}(m, m') = \emptyset)$. It remains to check the case it stops in step 4.
	The $d$-CFF property guarantees that if a block has not been modified, it is contained in at least one valid concatenation with $T[i] = T'[i]$, since there is a row $i$ that avoids all modified blocks. Therefore, this block is contained in $V_i$. For each row $i$ that contains a modified block, we have $T[i] \neq T'[i]$, so modified blocks are not contained in any $V_i$. Therefore $I$ consists of precisely the modified blocks. Thus, $|I| \leq d$, and the algorithm outputs $(1, I = \text{diff}(m, m'))$.  \qed
\end{proof}

The next theorem shows that when step 4 outputs $(0,I)$, Scheme 1 may give more information than required in Definition~\ref{mtss:def:MTSS}, as it may identify some unmodified blocks, even if not all.

\begin{theorem}\label{mtss:unmod}
	Consider a valid signature $\sigma$ generated by {\sc MTSS-Sign} for a message $m$ and key pair ($SK, PK$), and let $m'$ be a modified message with $|\text{diff}(m, m')|$ $> d$. Then, {\sc MTSS-Verify}($m', \sigma, PK$) = $(0, I)$, and for any $i \in \{1, \ldots, n\} \setminus I$, block $m[i]$ is guaranteed to be unmodified.
\end{theorem}
\begin{proof}
	This case will lead {\sc MTSS-Verify} to step 4, and since $|\text{diff}(m, m')|>d$, the output will be $(0,I)$. Any block in $\{1, \ldots, n\} \setminus I$ is guaranteed to be part of matching row $i$, and must be unmodified, even though not every unmodified block will necessarily be placed in $\{1, \ldots, n\} \setminus I$. \qed
\end{proof}	

Scheme 1 has been proposed in~\cite{mtss:thaisIPL}. One of our main contributions here is to add correcting capability to $d$-MTSS. We require the size of each block to be upper bounded by a value $s$ that is small enough that guessing each of the (up to) $d$ modified blocks is ``computationally feasible". Basically, by brute force we compute up to $O(d2^s+n)$ hashes to accomplish the modification correction (see the algorithm under Scheme 2).
We use the indices of the modified blocks in $I$ and do an exhaustive search to recover their original values.

\vspace{10pt}
\noindent
\textbf{Scheme 2: A $d$-MTSS with Correction Capability}
\vspace{5pt}

The scheme requires an underlying CDSS $\Sigma$, a public hash function $h$, and a $d$-CFF($t,n$) matrix $\mathcal{M}$. It further requires that the message is divided into $n$ blocks of size at most $s$. The scheme has the three algorithms from Scheme 1, and additionally the algorithm below:

\begin{itemize}
	\item {\sc MTSS-Verify\&Correct}($m, \sigma, PK$): receives as input a message \\$m = (m[1], \ldots, m[n])$, a signature $\sigma = (\sigma', T)$ where $T = (T[1], T[2], \ldots, T[t], h^*)$, and a public key $PK$, \new{and} proceeds as follows.
	\begin{enumerate}
		\item Compute \emph{result} = {\sc MTSS-verify}($\new{m}, \sigma, PK$). If \emph{result} = $(0, X)$, then stop and output $(0,X)$, otherwise \emph{result} = ($1,I$). If $|I| = 0$ go to step 6, otherwise run steps $2-5$ for each $k \in I$. 
		\item Identify a row $i$ in $\mathcal{M}$ such that $\mathcal{M}_{i,k} = 1$ and $\mathcal{M}_{i,j} = 0$, for all $j \in I \setminus \{k\}$. 
		\item Compute $h_j = h(m[j])$ for all $j$ such that $\mathcal{M}_{i,j} = 1, j \neq k$, $i$ from step 2. Set \emph{corrected[k] = false}.
		\item For every possible binary string $b$ of size $\leq s$, proceed as follows:
		\begin{itemize}
			\item Compute $h_k = h(b)$.
			\item For $i$ obtained in step 2 and $1 \leq j \leq n$, compute $c_i$ as the concatenation of every $h_j$ such that $\mathcal{M}_{i,j} = 1$ and set $T'[i] = h(c_i)$.
			\item If $T'[i] = T[i]$ and \emph{corrected[k] = false}, set \emph{corrected[k] = true} and correct the block $m[k] = b$. 
			\item Else, if $T'[i] = T[i]$ and \emph{corrected[k] = true}, stop and output ($1, I, \lambda$).
		\end{itemize} 
		\item Return to step 2 with the next $k \in I$.
		\item Output $(1, I, m)$.
	\end{enumerate}
\end{itemize}

We note that the flag $corrected[k]$ is used to identify a possible collision of two different bit strings $b$ giving the same hash value $h(m[k])$. Since the correct block cannot be determined, we exit with $\lambda$ indicating failure.
The next proposition has details on the correctness of this algorithm.

\begin{proposition}\label{mtss:propcorr}
	Let $(m, \sigma)$ be a valid pair of message and signature produced by {\sc MTSS-Sign}($m, SK$), using a hash function $h$ and with $m = (m[1], \ldots, m[n])$. Let $m'$ be a message and let $I = \text{diff}(m, m')$ with $|I|\leq d$. If for every $k \in I$, $h(m[k])$ has no other preimage of size up to $s$, then \\{\sc MTSS-Verify\&Correct}($m', \sigma, PK$) = $(1, I, m)$.
\end{proposition}
\begin{proof}
	As seen in Theorem~\ref{mtss:verifyCorrect}, the set $I$ contains precisely the \new{indices} of the modified blocks, and Theorem~\ref{mtss:uniqueT} guarantees that such a row $i$ in step $2$ of the algorithm exists.
	Finally, if for every $k \in I$ the hash $h(m[k])$ has no second preimage, 
	then step $4$ of the algorithm computes a unique replacement for each modified block, and the algorithm outputs the corrected message in step $6$. \qed
\end{proof} 

Now we prove that when selecting a good hash function, we can always guarantee the correction of any up to $d$ modified blocks.

\begin{theorem}\label{mtss:theocollision}
	Consider Scheme 2 restricted to messages with blocks of size at most $s$, and such that $h$ is a hash function where no two inputs of size up to $s$ have the same hash value. Then, {\sc MTSS-Verify\&Correct}($m', \sigma, SK$) can always correct a message with up to $d$ modified blocks.
\end{theorem}
\begin{proof}
	Easily obtained by Proposition~\ref{mtss:propcorr} since no matter what is the value of the block, no other block of size up to $s$ can have the same image under $h$. \qed
\end{proof}	

Next, we show that the assumption of existence of such hash function is realistic, and in fact very easy to find.

\begin{proposition}\label{mtss:probcoll}
	Consider a family of hash functions $h:X\rightarrow Y$ where $|Y|=2^{l}$ and a subset of the inputs $S \subseteq X$ where $|S|\leq 2^{s+1}$. The probability that there is collision in $S$, i.e. the probability that there exists x,z in S with $h(x)=h(z)$, is approximately $\epsilon = 1- e^{-2^{2s-l+1}}$.  
\end{proposition}
\begin{proof}
	This comes from the fact that the probability of finding at least one collision after $Q$ hash calculations is approximately $1-e^{-\frac{Q^2}{2^{l+1}}}$, and in our application $Q = |S| = 2^{s+1}$. 
\end{proof}

In practice, we will set $2s \ll l$. This ensures via Proposition~\ref{mtss:probcoll} that $h$ is very likely to have the desired property of being injective on $S$, the set of binary strings with size at most $s$. In the unlikely event that $h$ fails this property, we try another hash function until we succeed. The expected number of attempts will be very close to 1.
For example, if we consider SHA-$256$ as the hash function, $|Y| = 2^{256}$, and for $s = 20$, the probability of a collision within the set of size at most $20$ is $\epsilon = 1- e^{-2^{40-256+1}} \approx 3.70 \times 10^{-68}$. Indeed, we experimentally verified that SHA-$256$ has no collisions for $s = 20$.

\begin{theorem}
	Consider Scheme 2 with tolerance level $d$, and let $m$ be a message split into $n$ blocks, each block of size at most $s$. Let $\sigma$ = {\sc MTSS-Sign}($m, SK$) and $m'$ be a message with $I = \text{diff}(m,m')$ with $|I|\leq d$.  Then, {\sc MTSS-Verify\&Cor-} {rect}($m', \sigma, PK$) returns $(1, I, m)$ and uses $O(n + d^2 \log n + d2^{s})$ hash computations.
\end{theorem}
\begin{proof}
	By choosing a suitable hash function, Theorem~\ref{mtss:theocollision} guarantees that {\sc MTSS-Verify\&Correct} returns $(1, I, m)$. 
	The algorithm starts with the location of the modified blocks. This step uses a $d$-CFF for which we can use a construction with $t = O(d^2 \log n)$ rows~\cite{mtss:PR}, and therefore a total of $n+t$ hash calculations are required to locate these modifications. 
	After locating up to $d$ modified blocks, we need to perform the following computations for each one of them. We compute every possible block of size up to $s$ and their corresponding hash values (total of $2^{s+1}-1$), 
	and according to the row of the CFF matrix, we compute a few extra hashes (in total not more than $n$, if storing $h_i$ instead of recomputing among different runs of of line 3). This gives a total of $O(d2^s + n)$ hash computations for the correction step. 
\end{proof}

\subsection{Comparing Scheme 2 with \new{Scheme E}}\label{section:comparison}

Let an ($l,n,D$)$_q$-ECC $\mathcal{C}$ be an error correcting code with minimum distance $D$, codewords of length $l$ that encode messages of size $n$ over an alphabet of size $q$. The alphabet must be so that it can distinguish each possible block considered in Scheme 2, so $q \geq 2^s$. For simplification, assume $c_m=(m,b_m)$ where $b_m$ is a tuple of $l-n$ letters (``check bits").
This code can correct $d = \lfloor \frac{D-1}{2} \rfloor$ errors (modifications in the message). 
Next, we describe signature and verification \new{using Scheme E}.
Using the same inputs as {\sc MTSS-Sign}, algorithm {\sc ECC-Sign}($m, SK$) do the following steps:
%\vspace{-7pt}
%\begin{enumerate}
%	\item 
	1) Compute $\sigma' = \Sigma$.{\sc Sign}($m,SK$);
	 %\item 
	2) Compute $c_m=(m,b_m)$ according to ECC $\mathcal{C}$ and return $\sigma = (b_m, \sigma')$.
%\end{enumerate}
Then, algorithm {\sc ECC-Verify\&Correct}($m, \sigma=(b, \sigma'), PK$) do the following steps:
%\vspace{-7pt}
%\begin{enumerate}
	%\item 
	1) Decode $c=(m,b)$ to $m'$ using $\mathcal{C}$;
	%\item 
	2) If $\Sigma$.{\sc Verify}($m', \sigma, PK$) = 1 then return $(1, \text{diff}(m, m'), m')$, else return $0$.
%\end{enumerate}

Since $\mathcal{C}$ can correct up to $d$ errors, if the signature $\sigma$ is valid ($\sigma'$ is authentic for $m'$ and $|\text{diff}(m, m')| \leq d$), then {\sc ECC-Verify\&Correct} will behave in the same way as {\sc MTSS-Verify\&Correct} and will return $(1, \text{diff}(m, m'), m')$. However, when {\sc ECC-Verify\&Correct} returns $0$, it does not distinguish between the two failing cases obtained by {\sc MTSS-Verify\&Correct}, namely:
\textbf{Case 1)} output $(0, -)$, which means the CDSS signature $\sigma'$ is not authentic;
\textbf{Case 2)} output $(0, I)$, which means $\sigma'$ is authentic and message $m$ differs from the signed $m'$ in more than $d$ blocks, and also if $|I| < n$, then $|\{1, \ldots, n\} \setminus I| > 0$ and we are sure of at least one unmodified block.
Therefore, while the \new{Scheme E} is an MTSS scheme according to Definition~\ref{mtss:def:MTSS}, it provides less information than Scheme 2.

A comparison of Scheme 2 with \new{Scheme E} shows they have similar compression ratios.
Indeed, we first note that given an $(l, n, D)_q$-ECC, it is possible to construct a $d$-CFF($l\cdot q, q^n$), with $d = \lfloor(l-1)/(l-D)\rfloor$ \cite{mtss:Niederreiter2007}. Then, we consider some families of error correcting codes, and for each family we restrict Scheme 2 to only use CFFs constructed using codes from this family. 
In the next table, we summarize the compression ratios obtained for 3 families of codes, but we omit the details of our calculations due to lack of space. 

\begin{center}
	\begin{tabular}{l|c|c|c}
		\hline
		code: & Hamming, $n=2^{2^r-1}$& Reed-Solomon, $n=q^{\frac{(q-1)}{d}+1}$& PR-code\cite{mtss:PR}\\ \hline
		Scheme 2 &    $\approx 2^{2^r-2r-2}$         &         $q^{\frac{(q-1)}{d}-1}$   &    $\Theta(\frac{n}{d^2 \log n})$                        \\
		\new{Scheme E} &    $\approx 2^{2^r-2r-1}$        &          $\frac{q^{\frac{(q-1)}{d}+1}}{2d+1}$   &               $O(\frac{(\log d) n}{d^2 \log n})$\\ \hline
	\end{tabular} 
	\end{center}

In conclusion, both schemes serve the same purpose with similar compression ratios, but Scheme 2 has several advantages. First, Scheme 2 gives more information in the failing case where ECC returns $0$, as discussed above. Second, Scheme 2 provides a non-correcting version of the verification algorithm ({\sc MTSS-Verify}) which in the case of unmodified messages is basically as time efficient as {\sc$\Sigma$.Verify} (see step 1 of {\sc MTSS-Verify}); in this case, \new{Scheme E} still needs to run a decoding algorithm, with complexity influenced by $q \geq 2^s$, where $s$ is the largest size of a block. Finally, Scheme 2 is part of a family of similar schemes presented here (Schemes 1-3), which can be more easily compared.

%%%%%%%%%%%%%%%%%%%%%%%%%%%%%%%%%%%%%%%%%%%%%%%%%%%%%%%%

\section{Security}\label{mtss:SecSecurity}

In this section, we present the security proof of $d$-MTSS for Schemes 1 and 2.
In order to do this, we need to check that the security of the hash function $h$ and the unforgeability of the underlying CDSS can together ensure unforgeability of $d$-MTSS. Note that although Scheme 1 appeared in~\cite{mtss:thaisIPL}, no security proof has been given in that paper. 
For the next proof we assume a \emph{collision-resistant} hash function, i.e. a hash function in which a collision cannot be efficiently found \cite{mtss:stinsonBook}.

When we consider $d$-MTSS using $d$-CFFs, a valid $(m, \sigma)$ as in Definition~\ref{mtss:validsig} implies that there exists $m'$ such that $\sigma$ is an output of {\sc MTSS-Sign}($m', SK$) and $|\text{diff}(m,m')| \leq d$. 
In the next theorem we suppose there is a valid $(m, \sigma)$ as a forgery to our scheme. We consider two types of forgery: a \emph{strong forgery} consists of $(m, \sigma)$ such that {\sc MTSS-Verify}($m,\sigma, PK$) = $(1,I), |I| \leq d$; a \emph{weak forgery} consists of $(m, \sigma)$ such that {\sc MTSS-Verify}($m,\sigma, PK$) = $(1,\emptyset)$.

\begin{theorem}\label{mtss:thoushallnotforge}
	Let X be a $d$-MTSS as described in Scheme 1 based on an existentially unforgeable CDSS $\Sigma$ and on a collision resistant hash function $h$. 
	Then, X is existentially unforgeable.
\end{theorem}
\begin{proof}(By contradiction)
	Suppose X is not existentially unforgeable. Then, there is an adversary $\mathcal{A}$ that, after $q$ adaptive queries to a signing oracle $\mathcal{O}$, obtains pairs $(m_1,\sigma_1), \ldots, \\(m_q,\sigma_q)$, with $\sigma_i = (\sigma'_i, T_i), 1 \leq i \leq q$, and with non\new{-}negligible probability,
	outputs a valid pair $(m, \sigma)$, with $\sigma = (\sigma', T)$, $T = (T[1], T[2], \ldots, T[t], h^*)$,  and $|\text{diff}(m, m_i)| > d$, for all $1 \leq i \leq q$. 
	
	We show that if such $\mathcal{A}$ exists, then we can build a probabilistic polynomial time algorithm $\mathcal{A}'$ which \new{has the following input and output}:
	
	\noindent
	\textbf{Input:} an existentially unforgeable CDSS $\Sigma$ and a collision-resistant hash function $h$, both with security parameter $\ell$.
	
	\noindent
	\textbf{Output:} either a existential forgery $(T, \sigma')$ of $\Sigma$ or a collision pair for $h$.
	
	Next we describe the steps of such $\mathcal{A}'$.
	
	\begin{enumerate}
		\item Simulate the probabilistic polynomial time adversary $\mathcal{A}$ forging an MTSS based on $\Sigma$ and $h$ using queries mentioned above. With non-negligible probability, this will produce a forgery $(m, \sigma = (\sigma', T))$ in X, as described above. Otherwise, return ``FAIL".
		\item If $T \neq T_i$, for all $1 \leq i \leq q$, $\mathcal{A}'$ presents ($T, \sigma'$) as a forgery in $\Sigma$, for $\sigma'$ the corresponding signature of $T$.
		\item If $T=T_i$, for some $1 \leq i \leq q$, first calculate $I$ by computing {\sc MTSS-Verify}($m,\sigma, PK$). Then, we have:
		\begin{itemize}
			\item In the case of weak forgery, we must have $I = \emptyset$, and so the final element of $T$ is $h(m)$ and the final element of $T_i$ is $h(m_i)$, and since $m \neq m_i$, $\mathcal{A}'$ presents ($m, m_i$) as a collision pair for $h$.
			\item Otherwise, we must have a strong forgery with $1 \leq |I| \leq d$. So, there exists $m'$ such that {\sc MTSS-Verify}($m',\sigma, PK$) = $(1, \emptyset)$ and $|\text{diff}(m,m')| \leq d$. %$\mathcal{A}$ can determine $I$ easily by running {\sc MTSS-Verify}. 
			Since $|\text{diff}(m, m_i)| > d$, there must be a block $m[k], k \in \{1,\ldots,n\} \setminus I$, that is considered valid in $m$ but $m[k] \neq m_i[k]$. Let $p$ be any row of the CFF matrix ${\cal M}$ with ${\cal M}_{p,k} = 1$. Because $T=T_i$, this implies $T[p] = h(c) = h(c') = T_i[p]$, for $c = h(m[j_1])||h(m[j_2]) || \ldots ||h(m[j_s])$ and $c' = h(m_i[j_1])||h(m_i[j_2]) || \ldots ||h(m_i[j_s])$, where $k \in \{j_1, j_2, \ldots, j_s\}$. If $h(m[k]) = h(m_i[k])$, then $\mathcal{A}'$ presents ($m[k], m_i[k]$) as a collision pair for $h$. Otherwise, $\mathcal{A}'$ presents ($c, c'$) as a collision pair for $h$.
		\end{itemize}
	\end{enumerate}
	The probability $p(\ell)$ that $\mathcal{A}'$ succeeds is the same as the probability that $\mathcal{A}$ succeeds, where $\ell$ is the security parameter. Whenever $\mathcal{A}'$ succeeds, we have either an existential forgery of $\Sigma$ or a collision for $h$, corresponding to steps 2 and 3, respectively.
	For each $\ell$, one of these steps is at least as likely to occur as the other, so it has probability at least $1/2$. So, for any security parameter $\ell$, use one of two algorithms (an existential forger for $\Sigma$, or a collision finder for $h$) that runs in probabilistic polynomial time and succeeds with probability at least $p(\ell)/2$. In other words, we can technically design two adversaries $\mathcal{A}'_1$ and $\mathcal{A}'_2$ based on $\mathcal{A}'$. $\mathcal{A}'_1$ forges $\Sigma$ by proceeding as $\mathcal{A}'$ but returning ``FAIL" if it falls in the case of step 3. Similarly, $\mathcal{A}'_2$ finds a collision for $h$ or returns ``FAIL" if it falls in the case of step 2. Then, either $\mathcal{A}'_1$ or $\mathcal{A}'_2$ will succeed with probability at least $p(\ell)/2$ for infinitely many $\ell$. This contradicts either the unforgeability of $\Sigma$ or the collision resistance of $h$. \qed
\end{proof}

\section{Using MTSS for Redactable Signatures}\label{mtss:SecRedactable}

Now we turn our attention to using MTSS in general and using similar algorithms to our proposed $d$-MTSS for redactable signature schemes.
Redactable and sanitizable signature schemes allow for parts of a message to be removed  or modified while still having a valid signature without the intervention of the signer.
Sanitizable signatures usually requires the existence of a semi-trusted party called the {\em sanitizer} who is entrusted to do the modifications and recomputation of the signature, in some schemes done in a transparent way 
%(no trace is left that a sanitizer has changed anything) and in others the sanitizer %(rather than the signer) is accountable for sanitization 
(see~\cite{mtss:PohlsThesis}). 
% for different properties of sanitizable signatures). 
Our proposed scheme does not deal with sanitizable, but rather with redactable signatures. 

Redactable signatures have been proposed in several variants also under the names of content extraction signatures~\cite{mtss:Steinfeld} and homomorphic signatures~\cite{mtss:Johnson}. 
Redactable signature schemes (RSS) ``allow removing parts of a signed message by any party without invalidating the respective signature''~\cite{mtss:Derler} and without interaction with the signer. In~\cite{mtss:Steinfeld}, the authors mention content extraction for privacy protection or bandwidth savings.
Suppose Bob is the owner of a document signed by Alice and does not want to send the whole document to a third verifying party Cathy but only some extracted parts of the document; however, Cathy still needs to verify that Alice is the signer of the original document. Alice is agreeable with future redactions when she signed the document. An example given in~\cite{mtss:Steinfeld} is that Bob has his university transcripts signed by the issuing university (Alice) and wants to submit the transcripts to a prospect employer (Cathy) without revealing some of his private information such as his date of birth. Cathy must be able to verify that the signature came from the university in spite of the redaction. In addition to the privacy application, content extraction can be used in a similar way when only part of a large document needs to be passed by Bob to Cathy for the purpose of reducing the communication bandwidth~\cite{mtss:Steinfeld}. 

The notion of redactable signatures we consider next is in line with our general definition of MTSS, but differs in that we add another algorithm called {\sc MTSS-Redact} and we require total privacy of the redacted parts. We give next a $d$-MTSS version that is redactable based on ideas similar to Scheme 1 but modifying it to guarantee total privacy of redacted blocks.
As mentioned in~\cite{mtss:PohlsThesis}, the scheme proposed in~\cite{mtss:thaisIPL} which was presented as Scheme 1 does not meet standard privacy requirements and in particular leaks the original message's hash value. The scheme we propose below addresses this issue by individually signing each of the hashes of the groups of blocks, and at the time of redaction of blocks, also redacting from the signature tuple any hashes involving concatenations of the modified blocks. To avoid more complex forms of forgery, that could take advantage of combining individual signatures of concatenations coming from different signed messages, we add the same random string to the various hashes of concatenations to link the individual parts that are signed at the same time; in addition, to avoid reordering of blocks within the same message via reordering the groups, we add a group counter to these hashes before signing. 

In the description below\new{,} a redaction is represented by the symbol $\blacksquare$.

\vspace{10pt}
\noindent
\textbf{Scheme 3: A Redactable $d$-MTSS with Privacy Protection}
\vspace{5pt}

The scheme requires an underlying CDSS $\Sigma$, a public hash function $h$, a $d$-CFF($t,n$) matrix $\mathcal{M}$ and a random number generator {\sc rand}. The algorithms are given next:
\begin{itemize}
	\item {\sc MTSS-KeyGeneration}$(\ell)$: generates a key pair ($SK, PK$) using algorithm \\$\Sigma$.{\sc KeyGeneration}$(\ell)$.
	\item {\sc MTSS-sign}($m, SK$): Takes as input  $SK$ and a message $m = (m[1], \ldots, m[n])$, and proceeds as follows.
	\begin{enumerate}
		\item Calculate $h_j = h(m[j]), 1 \leq j \leq n$. Compute a random string $r=${\sc rand}(). 
		\item For each $1 \leq i \leq t$, compute $c_i$ as the concatenation of all $h_j$ such that $\mathcal{M}_{i,j} = 1, 1 \leq j \leq n$, and set $T[i] = h(c_i)||r||id(i,t+1)$, where $id(i,t+1)$ encodes the numbers $i$ and $t+1$.
		\item Compute $h^* = h(m)$, set $T[t+1]=h^*||r||id(t+1,t+1)$ and set $T = (T[1],  \ldots, T[t+1])$.
		\item Calculate $\sigma'[i]= \Sigma$.{\sc Sign}($T[i], SK$), for each $1 \leq i \leq t+1$ and set $\sigma'=(\sigma'[1], \sigma'[2], \ldots, \sigma'[t+1])$. Output signature $\sigma=(\sigma',r)$.
	\end{enumerate}
	\item {\sc MTSS-verify}($m, \sigma, PK$): takes as input a message $m = (m[1], \ldots, m[n])$, a signature $\sigma = (\sigma', r)$, and a public key $PK$, \new{and} proceed as follows. 
	\begin{enumerate}
		%\item Ensure that $\Sigma$.{\sc Verify}$(T, \sigma', PK) = 1$, stop and output $0$ otherwise.
		\item Check if $\Sigma$.{\sc Verify}$(h(m)||r||id(t+1,t+1),\sigma'[t+1],PK)=1$. Stop and output $(1, \emptyset)$ if that is the case; continue otherwise.
		\item Use $\mathcal{M}$, $m$, $r$ and do the same process as in steps 1-3 of {\sc MTSS-sign} to produce tuple  $T'=(T'[1], \ldots, T'[t])$.
		\item For each $1 \leq i \leq t$ such that $\Sigma$.{\sc Verify}$(T'[i],\sigma'[i],PK)=1$, compute the set of indices of intact blocks $V_i = \{j: \mathcal{M}_{i,j} = 1\}$ and do $V=V\cup V_i$. Compute $I = \{1, \ldots, n\}\setminus V$. Output $(1,I)$ if $|I| \leq d$; output $0$ otherwise.
	\end{enumerate}
	\item {\sc MTSS-redact}($m, \sigma, R$):  takes as input a message $m = (m[1], \ldots, m[n])$, a signature $\sigma = (\sigma', r)$, and a set $R\subseteq \{1,\ldots,n\}$ of blocks to be redacted, with $|R| \leq d$, and proceeds as follows.
	\begin{enumerate}
		\item If $R=\emptyset$, then stop and output $(m,\sigma)$.
		\item Create copies of the message and signature: $\overline{m}=m$, $\overline{\sigma}={\sigma}$, so that $\overline{\sigma}=(\overline{\sigma}',r)$.
		\item Set $\overline{\sigma}'[t+1]=\blacksquare$.
		\item For each $j\in R$: set $\overline{m}[j]=\blacksquare$ and for each index $i$ such that $\mathcal{M}_{i,j}=1$ set $\overline{\sigma}'[i]=\blacksquare$
		\item Output $(\overline{m},\overline{\sigma}=(\overline{\sigma}',r))$.
	\end{enumerate}
\end{itemize}

The correctness of the first three algorithms follows similar reasoning using the CFF properties as argued in Theorem~\ref{mtss:verifyCorrect}, taking into account the different approach used here where $t$ CDSS signatures and $t$ CDSS verifications are required.
The redaction offers total privacy in the sense of information theory, as no information related to the redacted blocks is kept in the signature, which are erased in \new{steps} \new{3 and 4} of {\sc MTSS-redact}.
The redacted blocks and redacted parts of the signature will only affect the verification of the parts of the signatures involving redacted blocks; all other block \new{indices} will be indicated as unmodified in the output of line 3 of {\sc MTSS-verify}, as long as no more than $d$ blocks have been redacted.

In the next theorem, we establish the unforgeability of this scheme.

\begin{theorem}
	Let X be a redactable $d$-MTSS  given in Scheme 3 based on an existentially unforgeable CDSS $\Sigma$, on a collision-resistant hash function $h$, \new{and on a random Oracle.}
	Then, X is existentially unforgeable.
\end{theorem}
{\em Sketch of the proof.}
Similar reasoning as in Theorem~\ref{mtss:thoushallnotforge} can be used to argue that the unforgeability of $\Sigma$ and the collision resistance of the hash function are sufficient to protect against forgery at the level of individual parts of the tuple signature. However, we need to rule out more complex forgery attempts that could try to combine individual signatures in the signature tuple in different ways as well as different blocks from different messages. The use of a common random string in the creation of each $T[i]$ in lines 2 and 3 of algorithm {\sc MLSS-Sign} prevents an adversary from trying to combine message blocks and signature parts of $(m_1,\sigma_1), \ldots, (m_q,\sigma_q)$ coming from different calls to the \new{random} Oracle to create a new valid signature since, by definition, two different calls to the oracle will, \new{with high probability,} not produce the same random string. Furthermore, the concatenation of the encoding of the test index and the number of tuple elements, $id(i,t+1)$, within each tuple position $T[i]$ in line 2 and 3, \new{makes it highly unlikely that blocks are} added, removed or reordered, \new{since this would amount to finding a collision in the hash function. \qed} %Therefore, by making it very unlikely reordering of blocks within a message, and combinations of blocks from different messages, we fall into the case of forgeries of similar types as dealt with in the proof of Theorem~\ref{mtss:thoushallnotforge}. 

%\noindent
%\new{{\em \textbf{Remark}.} The random Oracle hypothesis can be weakened, as it is sufficient to draw numbers from a sequence that does not repeat numbers, or repeats with negligible probability.}

\section{Implementation Issues}\label{mtss:implementation}

When implementing Schemes 1-3, there are a few details that need to be considered regarding efficiency, security, and even flexibility on the inputs or outputs of the algorithms. For example, we consider a message divided into blocks that we represented as a sequence, but blocks may be represented using different data structures depending on the application and the type of data. Pohls~\cite{mtss:PohlsThesis} discusses sequence, set and tree data structures for blocks, which could be employed.

When signing a message using {\sc MTSS-Sign} as described in Schemes 1-3, we could consider not  hashing every block before concatenation, \new{e}specially if the sizes of the blocks are small enough that a hash function will end up increasing the size of the data being concatenated. This approach needs to be carefully considered, since a simple concatenation does not insert a delimitation on where the individual blocks start or end, and therefore blocks with a shared suffix or prefix may lead to wrong identifications of valid/invalid. To be safe, we can hash each block before concatenating them as presented in our schemes, or ensure the concatenations use delimiters to separate blocks, \new{or require blocks of fixed size.}

The correction algorithm {\sc MTSS-Verify\&Correct} computes the set of indices of modified blocks $I$ and tries to correct these blocks to their original value. If there is at least one position where two different blocks match the hash of the original one (a second preimage on the hash function), the algorithm aborts the correction process and outputs an empty string $\lambda$ to represent correction error. As seen in Proposition~\ref{mtss:probcoll} and the discussion that follows, it is very unlikely for such event to occur, and we could always choose another hash function with no collisions for certain block sizes. However, if some specific application is not flexible in the choice of hash functions and such \new{an} event happens, an interesting approach can be to correct as many modified blocks as possible, and return some information regarding the ones with collisions (such as a fail message or even a list of possible values for those specific blocks). This approach allows for partial correction of modified data, that may be interesting for some applications. Moreover, if we already know that the chosen hash function has no collisions for blocks of size up to $s$, we do not need to run step 4 for every possible block of size $\leq s$, and we could stop this loop as soon as we find the first match, improving the efficiency of the method.

Regarding the multiple hash computations that happen in the correction algorithm, one could consider some improvements. Note that we already compute the hashes of all unmodified blocks when we do the call of {\sc MTSS-Verify} to obtain the set $I$, so these hashes of unmodified blocks could be reused in step 2. Moreover, we repeat for every modified block in $I$ the same process of computing all blocks of size $\leq s$ and their corresponding hash values. For the cases where we have a big set of modified blocks (big $d$), one may consider to pre-compute all these values, in case the application can handle the storage that this will require.

\section{Parameter Relations}\label{mtss:param}
For MTSS with modification location only (Scheme 1), a message $m$ is split into $n$ blocks of no specific size, and a location capability parameter $d$ needs to be decided. These parameters are used to construct a $d$-CFF($t,n$) with number of rows $t = \Theta(d^2 \log n)$ if using~\cite{mtss:PR}, $t = d \sqrt{n}$ if using~\cite{mtss:PPS}\new{, and $t = q^{2}$ when using~\cite{mtss:erdospoly}, for any prime power $q$ and positive integer $k$ that satisfy $q \geq dk+1$}. The signature size is impacted by $d$ and $n$, since there are $t+1$ hashes as part of the signature.  
Therefore, while smaller sizes of blocks and larger $d$ give a more precise location of modifications, they also increase the size of the signature.

Now consider the exact same message $m$, but for the case of an MTSS with correction capabilities (Scheme 2). This scheme requires that the blocks have a small size of up to $s$ bits, which implies that the very same message $m$ now has many more blocks $n' \gg n$. A larger number of blocks directly increases the size of the signature. But now, the $d$ that was enough to locate the modifications before may not be enough anymore, since modifications that were in one big block before now may be spread over several small blocks. When locating the modifications, the algorithm aborts the process if the number of modified blocks is larger than the expected location capability $d$, which may cause the scheme to fail to correct more often if this value is not increased as $n'$ increases.

The size of the signature $\sigma$ in Schemes 1 and 2 is $|\sigma| = |h(.)|\times(\new{t} + 1)+|\sigma'|$, and in Scheme 3 is $|\sigma| = |\sigma'|\times(\new{t} + 1)+|r|$, for \new{$\sigma'$ a \new{classical} digital signature, $r$ a random bit string, $h(.)$ a traditional hash function, and $t$ the number of rows of a $d$-CFF($t,n$)}. The input consists of a message of size $|m| = N$ in bits, divided into $n$ blocks, each of size at most $s$, so $N \approx n \times s$. 
In \new{Scheme 2,} given $N$, we need to wisely choose $s$ so it is small enough to allow corrections of blocks, while guaranteeing $n$ is small enough to have a reasonable $|\sigma|$. We \new{cannot} expect our signature to be as small as the ones from \new{classical} digital signature schemes, since we require extra information in order to be able to locate and correct portions of the data. In summary what we want is $n \times s \gg \new{|\sigma|}$.

\new{The} next example\new{s} show that even for small $s=8$, we still have a reasonably small signature. 

\noindent
\textbf{Example 1: }  $N = 1$ \new{GB} = $2^{30}$ bytes = $2^{30}2^3 = 2^{33}$ bits, $h =$ SHA-$256$, $s = \log |h| = \log 2^8 =8$ bits, $d = 4$, and we use RSA with \new{a} $2048$-bit modulus. Then we have $n = N/s = 2^{33}/2^3 = 2^{30}$ blocks, with $8$ bits each. \new{Consider the $d$-CFF($q^{2}, q^{k+1}$) from~\cite{mtss:erdospoly}, with $k=6, q = 25, t = 25^2, n = 25^7$.}
Now, since $|\sigma| = |h(.)|\times(\new{t} + 1)+|\sigma'|$ in Schemes 1 and 2, we have $|\sigma| = 2^8\times(\new{25^2}+ 1)+2048 = \new{162304}$ bits, which is $\new{20288}$ bytes, or $\new{\sim 20}$ \new{KB}.
\vspace{5pt}

\noindent
\textbf{Example 2:} For the same message and parameters as in Example 1, now we use a random bit string $r$ of size $128$ bits and create a signature as in Scheme 3. Since $|\sigma| = |\sigma'|\times(\new{t} + 1)+|r|$, we have $|\sigma| = 2048\times(\new{25}^2+ 1)+128 = \new{1282176}$ bits, which is $\new{160272}$ bytes, or $\sim \new{160}$ \new{KB}. 

For Scheme 3, small blocks like in Example 2 are not required, so we can get a much smaller signature. The choice of $s$ and $n$ in this case depends on the application, and on whether we wish to correct non-redactable blocks. %\new{In the case where we want to correct non-redactable blocks, note that both non-redactable and redactable blocks are all contributions to the same maximum modification tolerance $d$.}

\section{Conclusion} \label{mtss:conclusion}
We introduce modification tolerant signature schemes (MTSS) as a general framework to allow localization and correction of modifications in digitally signed data. We propose a scheme based on $d$-CFFs that allows correction of modifications, and a variation without correction that gives redactable signatures. The presented schemes are provably secure and present a digital signature of size close to the \new{best known lower bound} for $d$-CFFs. 

Interesting future research includes an implementation of the schemes proposed here, with practical analysis of parameter selection for specific applications. In addition, new solutions for MTSS beyond $d$-CFFs can also be of interest. The $d$-CFF treats every block the same and allows for any combination of up to $d$ blocks to be  modified.  %\todo{to complete the security proof..}
Specific applications can have different requirements about what combinations of blocks can be modified by specifying a different authorized modification structure. Moreover, other \new{hypotheses} about more likely distribution of modifications throughout the document can be used for efficiency\new{;} for example, modified blocks may be concentrated close to each other. One idea in this direction is to use  blocks with sub-blocks for increasing granularity of modification location and to aid correction. This would involve matrices with a smaller $d$ for the bigger blocks and a larger $d$ for sub-blocks of a block, picking these parameters with the aim of decreasing $t$ and consequently signature size. It is out of the scope of this paper to go into detailed studies of other types of modification location matrices, which we leave for future work.

\end{document}